\documentclass[final,5p,twocolumn]{elsarticle}

\usepackage{graphicx}          
\usepackage[dvips]{epsfig}    

\usepackage{xcolor}

\usepackage{amssymb}
\usepackage{amsthm}

\usepackage{amsmath}          

\usepackage{algorithmic}

\journal{Control Engineering Practice}

\begin{document}

\begin{frontmatter}

\title{Reduced-order asymptotic observer of nonlinear friction for precise motion control}

\author{Michael Ruderman}

\ead{michael.ruderman@uia.no}

\address{University of Agder, 4879 Norway}

\begin{abstract}                          
Nonlinear friction has long been, and continues to be, one of the major challenges for precision motion control systems. A linear asymptotic observer of the motion state variables with nonlinear friction uses a dedicated state-space representation of the dynamic friction force (including pre-sliding) \cite{ruderman2023robust}, which is robust to the noisy input signals. The observer implements the reduced-order Luenberger observation law, while assuming the output displacement is the only available measurement. The latter is relevant for multiple motion control applications with use of encoder-type sensors. The uniform asymptotic stability and convergence analysis of the proposed observer are elaborated in the present work by using the Lyapunov function-based stability criterion by Ignatyev and imposing parametric constraints on the time-dependent eigenvalues of the system matrix to be always negative real. A design procedure for assigning a dominant, and thus slowest, real pole of the observer is proposed. Explanative numerical examples accompany the developed analysis. \textcolor[rgb]{0.00,0.00,0.00}{In addition, a thorough experimental evaluation is given for the proposed observer-based friction compensation which is performed for positioning and tracking tasks. The observer-based compensation, which can serve as a plug-in to a standard feedback controller, extends a PID feedback control that is optimally tuned for disturbance suppression. The experimental results are compared with and without the plugged-in observer-based compensator.}
\end{abstract}


\end{frontmatter}

\newtheorem{thm}{Theorem}
\newtheorem{lem}[thm]{Lemma}
\newtheorem{prop}[thm]{Proposition}
\newtheorem{clr}{Corollary}
\newdefinition{rmk}{Remark}
\newproof{pf}{Proof}

\section{Introduction}
\label{sec:1}

Observing the dynamic states \textcolor[rgb]{0.00,0.00,0.00}{(in contrast to stationary process variables)} in an actuated motion system can be essential for various reasons, especially when the relative output displacement is the only measured quantity and a high-precision motion \textcolor[rgb]{0.00,0.00,0.00}{(i.e., with a micro- and nano-meters range of accuracy)} is required. The classical Luenberger-type observers \cite{luenberger1971} require a linear state-space realization of the system dynamics, while a nonlinear observer design problem is often \textcolor[rgb]{0.00,0.00,0.00}{realized by solving a set of first-order PDEs} (partial differential equations) \cite{kazantzis1998} and requires an appropriate nonlinear coordinates transformation for the state variables and respective flows. 
Restrictive conditions on the flow maps and state transformations can prevent application of the (more or less) standard observer design methodologies 
to such physical effects as a nonlinear kinetic friction. One of the main problems related to the observation of nonlinear friction states is that they cannot be properly considered as functions of the time argument and possess, moreover, a memory mechanism resulting in hysteresis. Such inherent problems with friction \textcolor[rgb]{0.00,0.00,0.00}{in the controlled motion systems, see e.g., \cite{ruderman2023book}, are often pronounced} less in the temporal than in the spatial domain, especially at the onset, stop, and reversals of a relative displacement. T\textcolor[rgb]{0.00,0.00,0.00}{his fact does not allow to effectively use also the so-called disturbance observers} (DOBs) (see \cite{ohishi1987dob} for an original work and \cite{shim2009} for analysis of its robust stability in the closed-loop), since the friction force dynamics \cite{AlBender2008} makes the related consideration in Laplace domain less applicable. Recall, that no Laplace transformation exists for general dynamic nonlinear operators with memory. \textcolor[rgb]{0.00,0.00,0.00}{Still, observation (correspondingly dynamic estimation) of the friction forces remains often essential for the control tasks. The friction state variables cannot be measured properly (see, e.g., \cite{harnoy2008}) and, at the same time, cannot be effectively mitigated by the standard feedback control techniques, also when including an integral action, see \cite{ruderman2025} for a detailed analysis. Also the Kalman filtering \cite{Kalman60}, understood to have a similar structure as the Luenberger observer while augmented by zero-mean stochastic process and measurement noise, see e.g. \cite[chapter~11.4]{balasubramanian1989}, \cite[chapter~9.3]{antsaklis2009}, cannot be effectively used for estimating the nonlinear friction. This is mainly due to a flawed assumption that stochastic process noise can capture the pre-sliding friction transitions.
}

Several works, to be mentioned here, were explicitly addressing observation of the friction force in the past. In \textcolor[rgb]{0.00,0.00,0.00}{an earlier work} \cite{olsson1998}, the authors summarized modeling and observer-based compensation of the dynamic friction, while using the co-called LuGre friction model \cite{DeWit1995} which has one internal nonlinear dynamic state (i.e., the one not measurable and without a direct physical interpretation). However, the proposed observer requires the relative velocity to be the measurable quantity \textcolor[rgb]{0.00,0.00,0.00}{which is fed back}. \textcolor[rgb]{0.00,0.00,0.00}{Moreover, the LuGre model turned out \textcolor[rgb]{0.00,0.00,0.00}{(over time of further studies and investigations)} to be sensitive to implementation, often causing instabilities depending on the numerical solver in use. This appeared crucial in case of the noisy signals (i.e., the used relative velocity and observation error) and restrictive as for the class of admissible input signals. In particular, it is not allowing for non-continuous (stepwise) transitions or 'too high' velocity rates (i.e., accelerations), see \cite{barahanov2000} for details.} Similar issues with the use of the LuGre friction model, when constructing a passivity-based friction observer for compensation in the control loop, were also addressed in \cite{freidovich2009lugre} along with the real-time implementation aspects and experiments. A friction compensation approach by using a reduced-order state observer was also proposed in \cite{Mallon2006}. The approach relies on the estimation of relative velocity by means of a reduced-order asymptotic observer, which is then used for the modeled set-valued dry friction. This scheme does not involve, however, any dynamic state of the pre-sliding friction and, this way, does not observe the state of a friction force but relies rather on its model-based computation. More recently, the work \cite{kim2019} proposed an model-free friction observer which, however, requires the measurement of the input torque (correspondingly force) and relies on the assumed torque balance and nominal system transfer characteristics in Laplace domain. \textcolor[rgb]{0.00,0.00,0.00}{Therefore, it has similar to the mentioned above issues as the DOB-based approaches in time or frequency domain.} An observer which incorporates the friction dynamics into the state equations and uses a Luenberger-type estimation scheme was proposed in \cite{ruderman2015}, while also assuming the measured relative velocity. Despite a not fully elaborated convergence analysis, the approach showed remarkable performance in compensating the nonlinear friction and brought the position control error to the level of \textcolor[rgb]{0.00,0.00,0.00}{the sensor (i.e., encoder) resolution in the experiments}, thus achieving zero steady-state error. An advanced version of this observation scheme, which uses only the measured output displacement and improves the robustness in both the pre-sliding and gross sliding regimes, was presented in \cite{ruderman2023robust}, \textcolor[rgb]{0.00,0.00,0.00}{but still without evaluating the compensation}. This latter study represents a starting point of the present work. \textcolor[rgb]{0.00,0.00,0.00}{The main motivation for the work is threefold. First, the proposed observer does not use any other measured system signals than the output displacement, which can also be highly noisy. Second, the used pre-sliding friction modeling is robust to the noisy input signals, that makes an observer operational also during relatively short transient phases of the motion reversals. Third, the developed stability analysis allows a robust observer parameterization over the whole operation range of relative velocity and, moreover, guarantees that no transient state oscillations occur during the observer convergence. Finally, the work provides a detailed comparative experimental evaluation of position control with and without plug-in of the proposed observer.}

The rest of the paper is organized as follows. Section \ref{sec:2} describes the class of the motion systems with kinetic friction under consideration, while providing the basic assumptions, dynamic states, and pre-sliding regime of the nonlinear friction. The main results are given in Section \ref{sec:3}, where the system dynamics with friction is first written in an observable state-space form and then transformed into the regular form. Based thereupon, a reduced-order Luenberger-type observer is designed and its convergence analysis is established for the time-varying observer system matrix. Finally, the parametric conditions for a monotonic and robust observer convergence are derived and the sensor noise propagation is also addressed in context of a practical implementation. A dedicated experimental case study is provided in Section \ref{sec:4}, where a laboratory linear displacement actuator is first described, followed by the design of a benchmarking PID feedback control. \textcolor[rgb]{0.00,0.00,0.00}{The latter is optimized for robust suppression of disturbances, which are assumed to occur due to nonlinear and uncertain friction}. Both, the positioning and tracking control experiments are reported for the optimal PID control and the same PID control augmented by the proposed asymptotic observer. Brief concluding remarks are given at the end of the paper in Section \ref{sec:5}.

\section{Motion system with nonlinear friction}
\label{sec:2}

Without loss of generality, consider the following class of the dynamic motion systems
\begin{equation}\label{eq:2:1}
m \ddot{x}(t) + f\bigl(\dot{x}(t)\bigl) = u(t) - g(t),
\end{equation}
where $m$ is the overall system inertia (i.e., lumped mass in case of a translational motion), $f(\cdot)$ is the weakly known nonlinear kinetic friction, and $u$ is a driving force available for the control. Note that the kinetic friction, as a function of relative velocity, is upper-bounded for a bounded argument $\dot{x}$ and moreover Lipschitz, i.e., $\bigl|\dot{f}\bigr| < \Phi$, where $\Phi > 0$, is a known positive constant. The main functional properties of $f\bigl(\dot{x}\bigr)$ are defined below in Section \ref{sec:2:sub:1}. It is worth noting that the output displacement $x$ and input force $u$ are provided in the generalized coordinates so that \eqref{eq:2:1} is appropriate for describing equally a translational as well as a rotational motion. Also we note that an external and, at the same time, matched load force $g(\cdot)$ is known and, thus, can be pre-compensated, correspondingly subtracted from $u$ in case of the observer design. Our goal is first to analyze the dynamic friction behavior in a way that allows to construct a robust asymptotic observer $\tilde{f}(t) \rightarrow f(t)$. The latter uses only the output measurement $x(t)$ and enables to compensate for the friction quantity $f(\cdot)$ when $x(t) \rightarrow x_{\textmd{ref}}(t)$ is under the feedback control.

\subsection{Basic assumptions}
\label{sec:2:sub:1}

The following basic assumptions about the nonlinear friction force are made, cf. \cite{ruderman2023book}.

\begin{enumerate}[(i)]

   \item The overall friction force at steady-state is a superposition
    of the Coulomb and viscous friction forces. This fundamental principle is given by
    \begin{equation}\label{eq:2:1:1}
    f(t) = F_c(t) + F_v(t)
    \end{equation}
    and is in line with several established approaches for modeling the kinetic friction
    with lumped parameters, cf. e.g., \cite{armstrong1994,AlBender2008}.
    Important to notice is that the superposition can be temporary
    lost during the fast dynamic transients, especially \textcolor[rgb]{0.00,0.00,0.00}{at higher system accelerations}.     
    
    \item The Coulomb friction force at steady-state 
    depends on the direction of the relative displacement only, i.e.,
    \begin{equation}\label{eq:2:1:2}
    F_c = C_f \, \mathrm{sign} ( \dot{x} ),
    \end{equation}
    and is parameterized by the Coulomb friction coefficient $C_f > 0$. The latter unites the dimensionless friction coefficient, often denoted by $\mu$, and the normal force $F_N$ acting on the frictional body \textcolor[rgb]{0.00,0.00,0.00}{and assumed here to be constant}, i.e., $C_f = \mu F_N$. 
    
    \item The viscous friction force depends on the relative velocity only and is given by \begin{equation}\label{eq:2:1:3} 
        F_v = \sigma \, \dot{x}.
        \end{equation}
        The linear friction law \eqref{eq:2:1:3} is well known also as a viscous damping, where $\sigma > 0$ is the viscous friction (or linear damping) coefficient.
    
    \item The Coulomb and viscous friction coefficients are often weakly known. Since the overall friction process can be time-varying, due to environmental conditions like the temperature, wear effects, changes in lubrication or dust, and moreover it can be position- an load-dependent, one can assume
        \begin{equation}\label{eq:2:1:4}
        C_f, \sigma \neq \mathrm{const}, \quad 0 <  C_f < C_f^{\max}, \quad 0 <  \sigma < \sigma^{\max},
        \end{equation}
        where $C_f^{\max}, \, \sigma^{\max} > 0$ are the known upper bounds.
    
    \item The transients of the friction force $\dot{f}(t)$ become essential (for modeling and observation) either during the beginning, stop, and reversals of the relative motion or
        \textcolor[rgb]{0.00,0.00,0.00}{at the high rates of acceleration amplitude} $|\ddot{x}|$ \textcolor[rgb]{0.00,0.00,0.00}{(i.e., jerks)}. Otherwise, they can be neglected at steady-state motion processes.
        
\end{enumerate}

\subsection{Dynamic friction states}
\label{sec:2:sub:2}

In order to augment the steady-state friction states in \eqref{eq:2:1:1} with an appropriate dynamic behavior, consider (separately) the transient equations of the Coulomb and viscous friction force components, i.e., $\dot{F}_c$ and $\dot{F}_v$ respectively.

The nonlinear Coulomb friction undergoes \emph{pre-sliding} transitions before saturating at $\pm C_f$. This leads to
\begin{equation}\label{eq:2:2:1}
\dot{F}_c = \left\{%
\begin{array}{ll}
    \dot{x} \cdot \partial F_c / \partial x \, , & \hbox{if } |F_c| < C_f,  \\[1mm]
    0, & \hbox{otherwise,} \\
\end{array}%
\right.
\end{equation}
cf. \cite{ruderman2016,ruderman2017,ruderman2017b}. Note that $\partial F_c / \partial x$ is at least once continuously differentiable for $t > t_r$ upon each motion reversal, i.e., when $\textrm{sign}\bigl( \dot{x}(t_r) \bigr) \neq \textrm{sign}\bigl( \dot{x}(t_r^+) \bigr)$ where $t_r$ is the time instant of a motion reversal and $t_r < t_r^+$ while $(t_r^+ - t_r)\rightarrow 0$. Since \eqref{eq:2:2:1} is generic for describing the pre-sliding friction curves in the coordinates of relative displacement and force, a particular form of the underlying modeling of the friction force $F_c(t)$ in pre-sliding becomes secondary. For more tribological details on the characteristic frictional behavior during the pre-sliding we refer to e.g., \cite{koizumi1984,AlBender2005}. At the same time, only those modeling approaches which have an unique feature $\partial F_c(t_r^+) / \partial x \rightarrow \infty$, cf. \cite{koizumi1984} can guarantee for the motion stop in \emph{finite time} when the system \eqref{eq:2:1} is autonomous, i.e. $g,u=0$ , see \cite{ruderman2017b} for details. 

For transient dynamics of the viscous friction term, one can assume the first-order linear behavior 
\begin{equation}\label{eq:2:2:2}
\dot{F}_v = \beta^{-1} (\sigma \dot{x} - F_v).
\end{equation}
The corresponding time constant is $\beta > 0$, while at steady-state \eqref{eq:2:2:2} reduces to \eqref{eq:2:1:3}. The dynamic behavior \eqref{eq:2:2:2} captures the so-called frictional lag, see e.g., \cite{armstrong1994,AlBender2008} for details. Note that the time constant $\beta$ is usually relatively low and can be uncertain depending \textcolor[rgb]{0.00,0.00,0.00}{on topology of contact surfaces} and lubrication conditions. If its nominal value is not directly identifiable, an arbitrary $0 < \beta \ll m\sigma^{-1}$ can be assumed to be sufficiently lower than the time constant of a corresponding \textcolor[rgb]{0.00,0.00,0.00}{nominal linear dynamic system. The latter would approximate} the friction term in \eqref{eq:2:1} by $f(\dot{x}) \approx \sigma \dot{x}$.

\subsection{Pre-sliding friction}
\label{sec:2:sub:3}

A principal shape of the dynamic Coulomb friction force during pre-sliding is illustrated in Figure  
\ref{fig:presliding} in the relative displacement coordinates. 
\begin{figure}[!h]
\centering
\includegraphics[width=0.7\columnwidth]{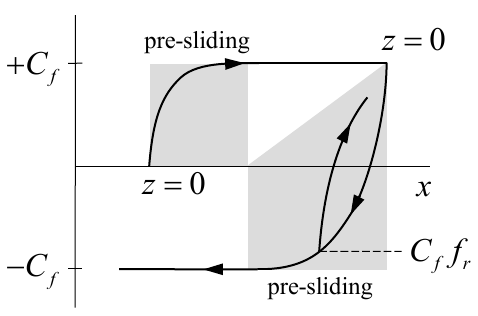}
\caption{Pre-sliding transitions of Coulomb friction force at motion beginning and reversals. Direction changes gives rise to hysteresis in friction.}
\label{fig:presliding}
\end{figure}
A pre-sliding region can be characterized after each motion reversal, or at the motion beginning, until the Coulomb friction force becomes saturated at $\pm C_f$. A pre-sliding displacement $z$ (also the so-called \emph{pre-sliding distance}) is understood as an internal state of the dynamic friction, cf. \cite{ruderman2023book}, and is captured by
\begin{equation}\label{eq:2:3:1}
z = s \int \limits^{t}_{t_{r}} \dot{x} \, dt,
\end{equation}
where $s > 0$ is a linear scaling factor used when modeling the $F_c(t)$ transition curves.    

We assume the following logarithmic form of the force transitions, which results in an increase in the corresponding hysteresis loop area proportional to the 2nd power of the pre-sliding distance
\begin{equation}\label{eq:2:3:2}
f_0(z)= z \bigl(1-\ln|z|\bigr),
\end{equation}
as suggested and justified in \cite{koizumi1984} for branching of the hysteresis friction force. Note that for \eqref{eq:2:3:2}, the pre-sliding distance state is well defined and monotonic only on the interval $z \in [-1, \, 0) \cup (0, \, 1]$, that requires an introduction of the scaling factor $s$. \textcolor[rgb]{0.00,0.00,0.00}{We use the thereupon-based modeling approach \cite{ruderman2017} which results in the  (normalized) pre-sliding friction map}
\begin{equation}\label{eq:2:3:3}
f_p(t)= \bigl| \mathrm{sign}(\dot{x}) - f_r \bigr| \, z
\bigl(1-\ln|z|\bigr) + f_r.
\end{equation}
This implies a pre-sliding friction memory $f_r := f_p(t_r)$ for each motion reversal at time $t_r$, cf. Figure \ref{fig:presliding}. 

Since the transient mapping \eqref{eq:2:3:3} is defined for the pre-sliding region only, i.e., for $|z| \leq 1$, the overall Coulomb friction force with dynamic transitions results in
\begin{equation}\label{eq:2:3:4}
F_c= \left\{%
\begin{array}{ll}
    C_f \, f_p(t) , & \hbox{if } |z|
\leq 1,  \\[1mm]
    C_f \, \mathrm{sign}\bigl(\dot{x}(t)\bigr), & \hbox{otherwise,} \\
\end{array}%
\right.
\end{equation}
cf. with \eqref{eq:2:1:2}. For the pre-sliding displacement state $z$, which is linked to the relative output displacement $x$ by the scaling factor $s$, \textcolor[rgb]{0.00,0.00,0.00}{one obtains}
\begin{equation}\label{eq:2:3:5}
\frac{\partial F_c}{\partial z} = - C_f
\bigl|\mathrm{sign}(\dot{x})-f_r \bigr| \, \ln|z|,
\end{equation}
\textcolor[rgb]{0.00,0.00,0.00}{and, this way, also $\partial F_c / \partial x$. Here recall that $\partial F_c / \partial x = s^{-1} \partial F_c / \partial z$, cf. with \eqref{eq:2:3:1}.} Since $z(t_r) = 0$, it can be seen that the instantaneous \emph{pre-sliding stiffness} $\partial F_c(t_r^+) / \partial x = \partial F_c(t_r^+) / \partial z \rightarrow \infty$, that is a crucial modeling feature for motion stops in a finite time, cf. \cite{ruderman2017b}.

It is worth noting that while other fairly suitable modeling approaches, like Dahl \cite{dahl1976}, generalized Maxwell-slip (GMS) \cite{AlBender2005b}, or modified Maxwell-slip \cite{ruderman2013}, can also be used for describing the pre-sliding hysteresis curves in $F_c(t)$, they fail to provide an infinite instantaneous stiffness, \textcolor[rgb]{0.00,0.00,0.00}{that is addressed above and required for the (physical) free  motion will stop in a finite-time.}

\section{Observation of motion states with friction}
\label{sec:3}

In order to design an asymptotic state observer, see \cite{luenberger1971} for basics, the $g$-disturbance compensated system \eqref{eq:2:1} with \eqref{eq:2:2:1} and \eqref{eq:2:2:2} can be transformed into the state-space form with the system matrix $A$ and the input and output coupling vectors $B$ and $C$, respectively. Introducing the state vector $w \equiv (w_1, w_2, w_3)^T : = (x, \dot{x}, f)^T$, the corresponding state-space model
\begin{equation}\label{eq:3:1}
\dot{w} = \underset{A} {\underbrace{
\left(%
\begin{array}{ccc}
0     & 1     & 0 \\
0     & 0     & -1/m \\
0     & \bigl(\partial F_c / \partial x +\sigma/\beta \bigr)  &  0
\end{array}%
\right) } } w + \underset{B} {\underbrace{
\left(%
\begin{array}{c}
  0 \\
  1/m \\
  0
\end{array}%
\right)
 } }
u
\end{equation}
with $C = (1,0,0)$ is obtained. Note that the system matrix has one \textcolor[rgb]{0.00,0.00,0.00}{state-dependent} term $\partial F_c / \partial x$ so that the system is time-varying, i.e., $A(t)$, correspondingly state-varying, during the pre-sliding. While $\partial F_c / \partial x \rightarrow \infty$ at $t_r$ due to $z=0$, for practical handling of the system matrices we assume  $\partial F_c / \partial x \rightarrow \kappa$ with $0 \ll \kappa < \infty$. Then, one can show that in both boundary cases $\partial F_c / \partial x = 0$ (i.e., motion is in gross sliding) and $\partial F_c / \partial x = \kappa$ (i.e., motion is in pre-sliding) the $(A,C)$-pair is observable in the Kalman sense for any positive $\kappa < \infty$. Since $\partial F_c / \partial x$ decreases monotonically towards zero \textcolor[rgb]{0.00,0.00,0.00}{with the progressing time} $t > t_r$, the same observability is valid on the whole interval $\partial F_c / \partial x \in [0, \, \kappa]$. This allows designing an asymptotic state observer in Luenberger sense \cite{luenberger1971} and, thus, obtain a states' estimate $\tilde{w}(t) \rightarrow w(t)$ for $t \rightarrow \infty$. An explicit analysis of the observer convergence with the time-varying system matrix is in Section \ref{sec:3:sub:2}.

\subsection{Reduced-order Luenbeger observer}
\label{sec:3:sub:1}

Since the measured output state $x$ does not need to be explicitly observed, a reduced-order observer \cite{luenberger1971} can be applied. This is improving the convergence properties of $\tilde{w}(t)$ and, most importantly, simplifying the design and analysis of the parametric conditions for tuning. Recall that for a reduced-order Luenbeger observer, the state-space model \eqref{eq:3:1} has to be transformed into the regular form
\begin{equation}
\left(%
\begin{array}{c}
  \dot{\bar{w}} \\
  \dot{y} \\
\end{array}%
\right) =
\left(%
\begin{array}{cc}
A_{11}     & A_{12}  \\
A_{21}     & A_{22}
\end{array}%
\right) \left(%
\begin{array}{c}
  \bar{w} \\
  y \\
\end{array}%
\right) + \left(%
\begin{array}{c}
  B_{\bar{w}} \\
  B_y \\
\end{array}%
\right) \, u,
 \label{eq:3:1:1}
\end{equation}
where $\bar{w} = w_1$ and $y=(w_2,w_3)^T$. Also recall that \eqref{eq:3:1:1} transfers an initial state-space representation into the regular form by providing a separation into the
measurable and unmeasurable states $\bar{w}$ and $y$, respectively. The resulting exclusion of the estimate $\bar{w}$ reduces the order of observer dynamics and, thus, complexity of the associated poles assignment. The reduced-order Luenberger observer is given by, cf. \cite{luenberger1971},
\begin{eqnarray}
\label{eq:3:1:2}
\dot{\tilde{y}} &=& (A_{22}-LA_{12}) \tilde{y} + (B_y -L B_{\bar{w}}) u + \\
\nonumber   & & (A_{21} - L A_{11} + A_{22}L -LA_{12}L) \bar{w},
\end{eqnarray}
where $L \equiv (L_1, L_2)^T$ is the vector of observer gains to be assigned. The latter is understood as an observer design and discussed in detail below in Sections \ref{sec:3:sub:2} and \ref{sec:3:sub:3}. The dynamic variable $\tilde{y}$ is the estimate vector of the unmeasurable system states, and since $\tilde{y}(t)$ is \textcolor[rgb]{0.00,0.00,0.00}{transformed} according to the right-hand-side of \eqref{eq:3:1:2}, a back transformation
\begin{equation}\label{eq:3:1:3}
(\tilde{w}_2, \tilde{w}_3)^T(t) = \tilde{y}(t) + L \bar{w}(t),
\end{equation}
is also required for obtaining the observed system states of interest, see \cite{luenberger1971} for details.

\subsection{Convergence analysis}
\label{sec:3:sub:2}

Since both $u$ and $\bar{w}$ are the bounded dynamic exogenous quantities entering the observer state equation \eqref{eq:3:1:2}, the only condition for the asymptotic stability of observer and thus for $\tilde{y}(t) \rightarrow y(t)$, is that the observer system matrix $\tilde{A} \equiv (A_{22}-LA_{12})$ is Hurwitz. We recall however that
\begin{equation}\label{eq:3:2:1}
\tilde{A} = \left(
                    \begin{array}{cc}
                    -L_1 &  - 1/m \\
                    \partial F_c / \partial x + \sigma / \beta - L_2 & 0 \\
                    \end{array}
                  \right)
\end{equation}
is time-dependent, and that in general its eigenvalues do not provide sufficient or necessary information about the stability guarantees. Only if the eigenvalues of $\tilde{A}(t)$ depend "slowly" on time, then an asymptotic stability is related to the time-dependent eigenvalues, cf. \cite{Ilchmann2006}.

For proving \textcolor[rgb]{0.00,0.00,0.00}{rigorously} the asymptotic stability of $\tilde{A}(t)$, we use the stability theorem by Ignatyev \cite{Ignatyev1997} for linear time-varying second-order differential equations of the form
\begin{equation}\label{eq:3:2:2}
\ddot{\psi} +  a_1(t) \dot{\psi} +  a_0(t) \psi = 0, 
\end{equation}
where $a_0(t)$, $a_1(t)$ are non-negative, continuous and bounded functions of the time.
\begin{thm}[cf. Ignatyev\cite{Ignatyev1997}]
\label{theorem:1} 
If the coefficients of \eqref{eq:3:2:2} satisfy the following conditions 
\begin{enumerate}[(i)]
  \item $\; \exists \; \alpha_1 > 0 \; \forall \; t \geq 0 \; : \quad \bigl|\dot{a}_0(t)\bigr|+ \bigl|a_1(t)\bigr| \leq \alpha_1$, 
  \vspace{0.5mm} 
  \item $\; \exists \; \alpha_2 > 0 \; \forall \; t \geq 0 \; : \quad 0 < \alpha_2 \leq \dot{a}_0(t) + 2 a_0(t) a_1(t)$, 
\end{enumerate}
then the differential equation \eqref{eq:3:2:2} is uniformly asymptotically stable. 
\end{thm}
\begin{proof}
\textcolor[rgb]{0.00,0.00,0.00}{The original proof in given} \cite{Ignatyev1997}.
\end{proof}
\vspace{0.5mm}

Introducing an auxiliary state variable $\Psi \in \mathbb{R}^2$ and transforming \eqref{eq:3:2:1} into the transpose companion matrix form $\tilde{A}^{\ast}$ so that
$$
\dot{\Psi} = \tilde{A}^{\ast}(t) \Psi,
$$ 
one obtains \eqref{eq:3:2:2} for $\Psi \equiv [\psi, \dot{\psi}]^\top$ with the coefficients
\begin{eqnarray}
\label{eq:3:2:3}
  a_0 &=& 1/m \, \bigl(\partial F_c / \partial x + \sigma / \beta -L_2), \\[0.3mm]
  a_1 &=& L_1.
\label{eq:3:2:4}  
\end{eqnarray}

\begin{prop}
\label{prop:1} 
The time-varying observer system matrix \eqref{eq:3:2:1}, and hence the observer \eqref{eq:3:1:2}, is uniformly asymptotically stable if the observer gains satisfy
\begin{eqnarray}
\label{eq:3:2:5}
  L_1 & > & 0, \\[0.3mm]
  L_2 & < & \sigma / \beta.
\label{eq:3:2:6}  
\end{eqnarray}
\end{prop}
\vspace{0.5mm}
\begin{proof}
The proof consists in proving the conditions (i) and (ii) of Theorem \ref{theorem:1} for the coefficients \eqref{eq:3:2:3}, \eqref{eq:3:2:4}, which describe the characteristic polynomial of the observer system matrix $\tilde{A}^{\ast}(t)$ and so also $\tilde{A}(t)$ given by \eqref{eq:3:2:1}. 

Taking the time derivative of \eqref{eq:3:2:3}, with respect to \eqref{eq:2:3:5} and \eqref{eq:2:3:1}, one obtains 
\begin{equation}\label{eq:3:2:7}
\dot{a}_0(t) = C_1 s \, \dot{x} / z \quad \hbox{with} \quad  C_1 = C_f \bigl| \mathrm{sign}(\dot{x}) - f_r \bigr| / m.  
\end{equation}
Since $|z| > 0$ after each motion reversal and a bounded $|\dot{x}|$ can be assumed for any stable and implementable motion system \eqref{eq:2:1} for all times $ t > 0$, a bounded $\dot{a}_0(t)$ can be guaranteed, while $C_1$ is upper-bounded by $2 C_f / m$ per definition. Together with \eqref{eq:3:2:4}, it fulfils the condition (i) of Theorem \ref{theorem:1}. Then, the condition (ii) of Theorem \ref{theorem:1} requires that there exists a positive constant $\alpha_2$ so that 
\begin{equation}\label{eq:3:2:8}
C_1 s \, \dot{x} / z + 2L_1/m \, \bigl( \partial F_c / \partial x + \sigma / \beta -L_2 \bigr) \geq \alpha_2.  
\end{equation}
Since the first summand on the left-hand-side of \eqref{eq:3:2:8} is always non-negative, it can be assumed to be zero \textcolor[rgb]{0.00,0.00,0.00}{for the worst case of still guaranteing stability}. For $\dot{x}=0$ and with $L_1 > 0$, it is sufficient to prove $\partial F_c / \partial x + \sigma / \beta -L_2 > 0$ for all times $t > 0$. The lowest value of the pre-sliding partial derivative is zero for $|z| > 1$ cf. \eqref{eq:2:3:4}, that leads to \eqref{eq:3:2:6}. This completes the proof.
\end{proof}

\subsection{Monotonic observer convergence}
\label{sec:3:sub:3}

While Proposition \ref{prop:1} provides a parametric criterion for stability of the reduced-order asymptotic observer of $(\tilde{w}_2, \tilde{w}_3)$, one is also inherently interested in how the observer gains \eqref{eq:3:2:5}, \eqref{eq:3:2:6} must be assigned to achieve a potentially uniform  and effective state estimation.

Since the nonlinear dynamic friction state constitutes the \textcolor[rgb]{0.00,0.00,0.00}{primary}  challenge for observation, the uniformity can be expressed in a set of parameters that are equally suitable for both pre-sliding and gross sliding regimes of a non-stationary relative motion. At the same time, the measure of the effectiveness of an asymptotic state estimation can be formulated in such a way that to avoid an oscillating transient dynamics and, therefore, aiming for a possibly monotonic convergence of the estimation errors $(\epsilon_2, \epsilon_3) = (\tilde{w}_2, \tilde{w}_3)- (w_2, w_3)$ for all times $t > 0$.

Evaluating the eigenvalues of \eqref{eq:3:2:1} yields 
\begin{equation}\label{eq:3:3:1}
\lambda_{1,2} = -\frac{1}{2} \Bigl(L_1 \pm \sqrt{L_1^2 +
4/m \bigl( L_2 - \partial F_c / \partial x - \sigma/\beta
\bigr)} \, \Bigr).
\end{equation}
It can be seen that the time-dependent eigenvalues and thus also the assignment of the $(L_1, L_2)$-pair are dictated by the expression under the square root of \eqref{eq:3:3:1}. For making use of analysis of the eigenvalues \eqref{eq:3:3:1}, we will further assume a finite stiffness of the pre-sliding transitions at each time $t_r$ of a motion reversal, cf. Section \ref{sec:3}.

\begin{prop}
\label{prop:2} 
Assuming a bounded pre-sliding stiffness $\partial F_c / \partial x \Bigl |_{t=t_r} \rightarrow \kappa$ with $0 \ll \kappa < \infty$ at motion reversals, the eigenvalues \eqref{eq:3:3:1} of the observer system matrix \eqref{eq:3:2:1} can be \textcolor[rgb]{0.00,0.00,0.00}{assigned always negative real}, with dominant pole at
\begin{equation}\label{eq:3:3:2}
\max \{ \lambda_{1,2} \} = - \sqrt{\kappa/m} \,\textcolor[rgb]{0.00,0.00,0.00}{ (\rho - 1)},
\end{equation} 
by choosing 
\begin{eqnarray}
\label{eq:3:3:3}
  L_1 & = & 2 \rho \sqrt{\kappa/m}, \\[0.3mm]
  L_2 & = & \sigma / \beta + \kappa \bigl(1 - \rho^2 \bigr),
\label{eq:3:3:4}  
\end{eqnarray}
where $\rho > 1$ is the design parameter.

\end{prop}
\vspace{0.5mm}

\begin{proof}
Requiring the poles of the observer system matrix to be real implies
\begin{equation}\label{eq:3:3:5}
L_1^2 + 4/m \bigl( L_2 - \partial F_c / \partial x - \sigma/\beta
\bigr) \geq 0.
\end{equation}
\textcolor[rgb]{0.00,0.00,0.00}{This leads}, \textcolor[rgb]{0.00,0.00,0.00}{when \eqref{eq:3:2:6} is fulfilled}, to   
\begin{equation}\label{eq:3:3:6}
\sigma / \beta + \partial F_c / \partial x - m L_1^2 / 4 \leq L_2 < \sigma / \beta.
\end{equation}
\textcolor[rgb]{0.00,0.00,0.00}{To guarantee that \eqref{eq:3:3:6} holds}, the gain condition
$$
L_1^2 \geq 4/m \, \partial F_c / \partial x
$$
must be fulfilled for all times $t > 0$. Taking the boundary value $\kappa$ instead of $\partial F_c / \partial x$ and introducing the design parameter $\rho > 1$ leads to \eqref{eq:3:3:3}, which is necessary (but not sufficient) for the eigenvalues \eqref{eq:3:3:1} to be real $\forall \: t > 0$.
For obtaining the sufficient parametric condition, dictated by $L_2$ after the $L_1$-value is determined, one rewrites \eqref{eq:3:3:5} as 
$$
L_2 \geq \sigma/\beta  + \kappa - L_1^2 \,  m/4, 
$$
\textcolor[rgb]{0.00,0.00,0.00}{and then by substituting $L_1$ from \eqref{eq:3:3:3} one obtains \eqref{eq:3:3:4}. The parametric condition \eqref{eq:3:3:4} is sufficient.} Finally, substituting \eqref{eq:3:3:3} and \eqref{eq:3:3:4} into \eqref{eq:3:3:1}, results in 
\begin{equation}\label{eq:3:3:7}
\lambda_{1,2} = - \rho \sqrt{\kappa/m} \pm \sqrt{\bigl(\kappa - \partial F_c / \partial x\bigr) / m}.
\end{equation}
During the motion reversals, the second square root term in \eqref{eq:3:3:7} becomes zero, and \eqref{eq:3:3:7} results in the most left double real pole. \textcolor[rgb]{0.00,0.00,0.00}{Since and once each} pre-sliding transition saturates, the second square root term in \eqref{eq:3:3:7} converges to $\sqrt{\kappa / m}$, \textcolor[rgb]{0.00,0.00,0.00}{that leads to the dominant pole to be as in \eqref{eq:3:3:2}}. This completes the proof.  
\end{proof}

\subsubsection*{Numerical illustration}

For highlighting the observer parameterization provided by Proposition \ref{prop:2}, consider the model parameters to be similar (but not all exactly the same) to those of the identified experimental system, cf. Section \ref{sec:4}, and assign the observer gains \eqref{eq:3:3:3}, \eqref{eq:3:3:4} once with $\rho = 0.98$ and once with $\rho = 1.02$. Note that the implemented numerical simulation of the system \eqref{eq:2:1}, transformed into \eqref{eq:3:1}, includes additionally a band-limited white noise of the output value $w_1(t)$, see Figure \ref{fig:simconv} (b), \textcolor[rgb]{0.00,0.00,0.00}{so at to make it closer to real (physical) systems}. For the input $u(t) = K v(t)$, a low-pass filtered noise signal is assumed, see Figure \ref{fig:simconv} (a), whose amplitude is scaled so that to induce the relative velocities \textcolor[rgb]{0.00,0.00,0.00}{within a feasible range. Note that the velocity's range is usually bounded for any type of a motion control application}. Numerical simulations with 10 kHz sampling rate were performed. The numerical parameter values used in the simulations are listed in Table \ref{tab:1}. 
\begin{table}[!h]
\global\long\def\arraystretch{1.5}%
\caption{\label{tab:1}Numerical simulation parameters}
\centering
\begin{tabular}{|p{1.8cm}|p{0.8cm}|p{0.8cm}|p{0.8cm}|p{0.8cm}|p{0.8cm}|}
\hline {Parameter} &{$K$}    & {$m$}    & {$\sigma$}    &{$C_f$}   &{$s$}\tabularnewline 
\hline {Value}     &{$3.28$} & {$0.538$}& {$21.1$}      &{$0.35$}  &{$3000$}\tabularnewline 
\hline \hline {Parameter} &{$\beta$}&{$\kappa$}&{$\omega_{co}$}&{$\rho_1$}&{$\rho_2$}\tabularnewline 
\hline {Value}     &{$0.001$}& {$8000$} & {$40$Hz}      & {$0.98$} &{$1.02$}\tabularnewline 
\hline
\end{tabular}
\end{table}

The observed friction state $\tilde{w}_3(t)$ is depicted versus the simulated one in Figure \ref{fig:simconv} (c). \textcolor[rgb]{0.00,0.00,0.00}{Once can recognize that for design parameter value $\rho = 0.98$, that violates the stability condition of Proposition \ref{prop:2}, the state estimation starts to diverge, although only after the first 4.3 sec}. On the contrary, the very slightly differing $\rho = 1.02$, \textcolor[rgb]{0.00,0.00,0.00}{which still fulfills} Proposition \ref{prop:2} and thus the stability criterion, leads to a robust monotonically converged estimate, while the output noise is propagated into the $\tilde{w}_3(t)$ state. Note that a long-term simulations (like e.g., $t > 60$ sec) with $\rho = 1.02$, or simulations with the further increased output noise, provide equally a stable and uniformly-converging friction estimate $\tilde{w}_3(t)$. For the sake of completeness, also the estimated velocity state $\tilde{w}_2(t)$ when $\rho = 0.98$ is compared with the simulated one $w_2(t)$ in Figure \ref{fig:simconv} (d).
\begin{figure}[!h]
\centering
\includegraphics[width=0.98\columnwidth]{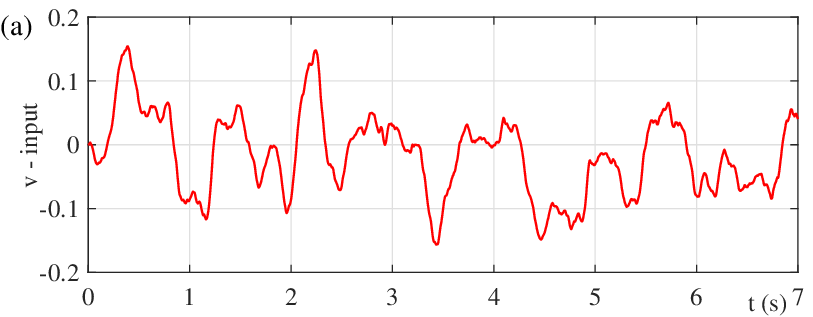}
\includegraphics[width=0.98\columnwidth]{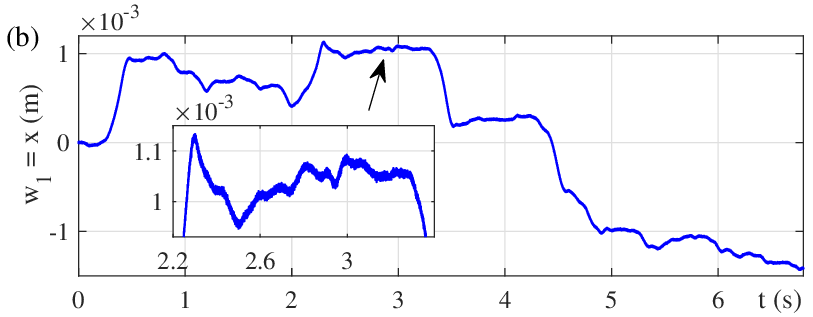}
\includegraphics[width=0.98\columnwidth]{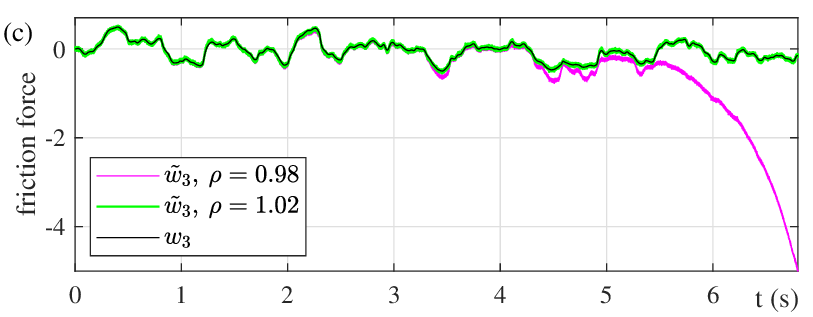}
\includegraphics[width=0.98\columnwidth]{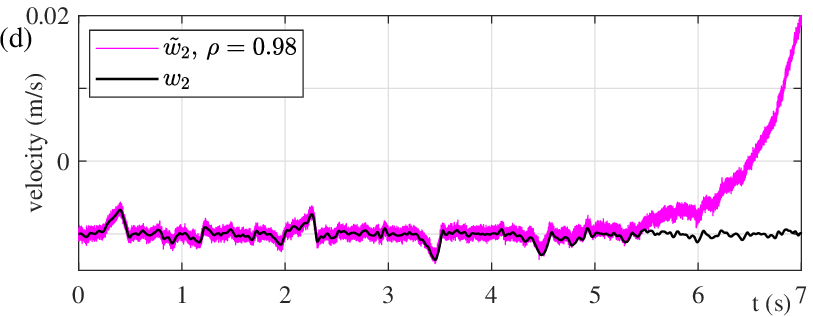}
\caption{Numerical example: band-limited random input signal $u(t)$ in (a), noisy output displacement state $w_1(t)$ used as the observer input in (b), simulated $w_3(t)$ versus observed  $\tilde{w}_3(t)$ friction state for two design parameters values $\rho = \{0.98, \, 1.02\}$ in (c), and simulated $w_2(t)$ versus observed  $\tilde{w}_2(t)$ velocity state for $\rho = 0.98$ in (d).}
\label{fig:simconv}
\end{figure}

\subsection{Propagation of measurement noise}
\label{sec:3:sub:4}

Since the measured output displacement is subject to some sensor noise $\eta(t)$, which can only be assumed to be norm-bounded (like e.g., $|\eta| < \textmd{const}$), \textcolor[rgb]{0.00,0.00,0.00}{but constitutes generally} a random process, it is of interest to analyze the noise propagation through the observer by assuming
\begin{equation}\label{eq:3:4:1}
\bar{w}(t) = x(t) + \eta(t).
\end{equation} 
Substituting \eqref{eq:3:4:1}, then zeroing the control input (i.e., $u=0$), and evaluating the observer equations \eqref{eq:3:1:2}, \eqref{eq:3:1:3} with the assigned design parameters, the dynamic friction state results in
\begin{eqnarray}
\label{eq:3:4:2}
  \dot{w}_3 &=& \kappa  \rho^2 w_2 + 2 \kappa \rho^3 \sqrt{\kappa/m} \, \bigl(x(t) + \eta(t) \bigr), \\[1mm]
  \tilde{w}_3 &=& w_3 - \kappa(\rho^2-1) \bigl(x(t) + \eta(t) \bigr). 
\label{eq:3:4:3}
\end{eqnarray}
Assuming, without loss of generality, an operation point $x=0$ and substituting the integrated \eqref{eq:3:4:2} into \eqref{eq:3:4:3} yields 
\begin{equation}\label{eq:3:4:4}
\tilde{w}_3 = \kappa  \rho^2 w_1 + 2 \kappa \rho^3 \sqrt{\kappa/m} \int \eta(t) dt - \kappa(\rho^2-1) \eta(t),
\end{equation} 
Even though an unbiased noise signal will be averaged through the second integral term on the right-hand-side of \eqref{eq:3:4:4}, the third right-hand-side term leads to a large propagation of 
$\eta(t)$ into the $\tilde{w}_3(t)$ state since $\kappa \gg 1$. Note that this is independent of the observer design, provided $\rho > 1$ is as by definition. This fact postulates that the friction state estimate will always, as long as $\eta \neq 0$ holds, contain the high-frequency noise components, which must be filtered out in postprocessing. A standard second-order linear low-pass filter described by
\begin{equation}\label{eq:3:4:5}
\ddot{\tilde{f}}(t) + 2 \omega_{co} \dot{\tilde{f}}(t)  +  \omega_{co}^2 \tilde{f}(t) = \omega_{co}^2  \tilde{w}_3(t),
\end{equation} 
with the cutoff frequency $\omega_{co} > 0$, is used for obtaining the friction force estimation $\tilde{f}(t)$. Note that $\omega_{co}$ can be customary determined taking into account both the evaluated sensor noise and the required reference value dynamics. \textcolor[rgb]{0.00,0.00,0.00}{This is especially with regard to the motion reversals and, thus, the expected friction transients in pre-sliding for the known $s$ and reference velocity values.}

\section{Experimental evaluation}
\label{sec:4}

The experimental system used \textcolor[rgb]{0.00,0.00,0.00}{for the control evaluation when compensating} for the nonlinear dynamic friction is the voice-coil-motor (VCM)-based linear displacement actuator, see Figure \ref{fig:expsetup}.  More technical details, including the system identification, can be found, e.g., in \cite{ruderman2022}. The available system signals are the input voltage $v(t)$ (in V) and the output relative displacement $x(t)$ (in m), while the real-time sampling rate of the control unit is set to 10 kHz. The nominal system dynamics is given by
\begin{equation*} \label{eq:4:1}
m \ddot{x}(t) = u(t) - G - f(t),
\end{equation*}
where the input force is $u= K v$, with $K = 3.28$ N/V, while neglecting the electro-magnetic time constant of the actuator which is $\approx 0.0012$ sec. The overall moving mass is $m=0.538$ kg, and the constant gravity term is $G = 5.27$ N. Note that the latter is feed-forward pre-compensated and, therefore, does not further affect the control design. The nonlinear friction $f$ is assumed to be weakly known, while the nominal parameters $C_f = 0.35$ N, $s = 3000$, $\sigma = 21.1$ N$\cdot$s$/$m, and $\kappa = 8000 \cdot s$ are identified from a series of the dedicated quasi-stationary experiments. \textcolor[rgb]{0.00,0.00,0.00}{Since an estimated $\beta$-value is too small comparing to the system sampling rate, the viscous dynamics \eqref{eq:2:2:2} is not explicitly considered and, hence, it is substituted by \eqref{eq:2:1:3} with the corresponding adjustments of the state observer \eqref{eq:3:1:2}. Note that it has a minor impact on the observer performance, since the Coulomb friction is the main source of uncertainties and is crucial for compensation in the motion control systems.}
\begin{figure}[!h]
\centering
\includegraphics[width=0.85\columnwidth]{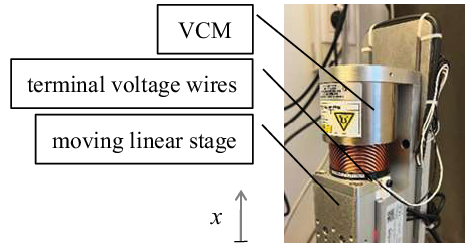}
\caption{Linear displacement actuator (laboratory view).} \label{fig:expsetup}
\end{figure}

\subsection{Optimally tuned PID control}
\label{sec:4:sub:1}

When approximating the motion system \eqref{eq:2:1} by an underlying transfer function (of complex Laplace variable $\lambda$)
\begin{equation}
\label{eq:linplant}
H(\lambda)=\frac{\phi}{\lambda(\tau \lambda +1)},
\end{equation}
with the overall system gain $\phi > 0$ and the dominant pole at $-1/\tau$, where $\tau > 0$ is the inertial time constant, then a standard PID control can be optimized in frequency domain with an objective of disturbance suppression. This is done based on the input disturbance sensitivity function
\begin{equation}\label{eq:Syd}
S(\lambda) = \frac{H(\lambda)}{1+H(\lambda)C(\lambda)}, 
\end{equation}
where the PID control is written in the following form
\begin{equation}\label{eq:Cpid}
C(\lambda) = \frac{k (\tau \lambda +1) (T_\mathrm{I}\lambda +
1)}{T_\mathrm{I}\lambda}.
\end{equation}
After cancelation of the plant system pole, the remaining control design parameters are the total control gain $k>0$ and the integrator time constant $T_\mathrm{I} > 0$.

By assigning the upper bound of $\bigl|S(\lambda)\bigr|$, as design specification, both control parameters are determined as 
\begin{equation}\label{eq:pidGain}
k = \frac{1}{\max |S|} 
\end{equation}
and 
\begin{equation}\label{eq:pidTime}
T_\mathrm{I} = \frac{\tan^2(\Phi)}{k \phi \sqrt{1 + \tan^2(\Phi)}},
\end{equation}
where the phase related argument $\Phi$ is computed for the open-loop crossover frequency $\omega_\mathrm{s}$ as
\begin{equation}\label{eq:L_omegas}
\Phi = \pi + \angle \bigl[ C(j \omega_\mathrm{s}) H(j \omega_\mathrm{s})  \bigr].
\end{equation}
The designed, this way, PID control \eqref{eq:Cpid} ensures that any matched disturbance $d(\lambda)$ will affect the controlled output $x(\lambda)$ as bounded by  
$$
\bigl |x(j \omega) \bigr | \leq \max |S| \cdot \bigl |d(j \omega) \bigr | \quad \forall \: \omega,
$$
where $\omega$ is the angular frequency and $j$ is the imaginary unit. More details on the above  given optimal design of a PID control for disturbance suppression can be found in \cite{ruderman2022sensitivity}.

\subsection{Comparison of experimental control results}
\label{sec:4:sub:3}

Two experimental scenarios are purposefully used. The first one is the positioning task for the reference profile $x_{\textmd{ref}}(t)$ depicted in Figure \ref{fig:expref} (a). The second one is the motion tracking of a smooth reference trajectory $x_{\textmd{ref}}(t)$ depicted in Figure \ref{fig:expref} (b). Note that the latter, which is an up-chirp sequence between 0.01 Hz and 3 Hz, induces the reference velocities in the range $|\dot{x}_{\textmd{ref}}| \in [0,\, 0.1]$ m/sec.  
\begin{figure}[!h]
\centering
\includegraphics[width=0.98\columnwidth]{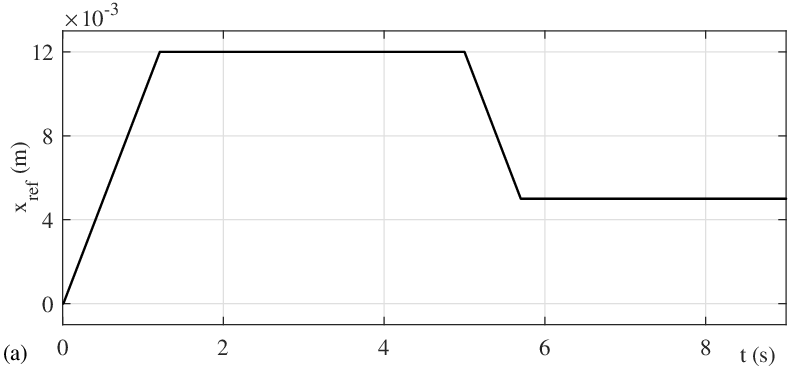}
\includegraphics[width=0.98\columnwidth]{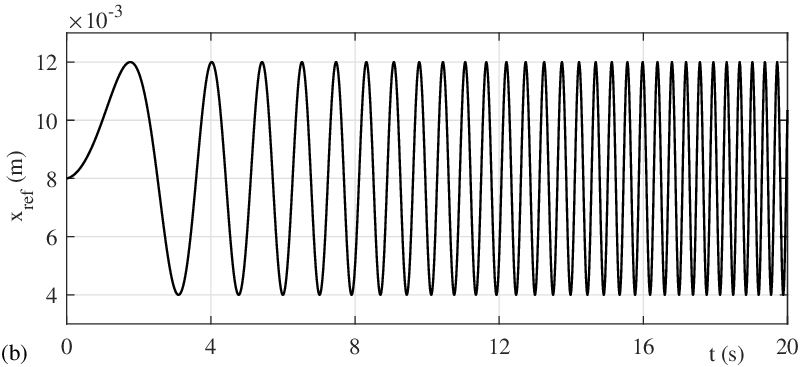}
\caption{Reference trajectories of the evaluated positioning (a) and tracking (b) control tasks.}
\label{fig:expref}
\end{figure}

The reference PID control, which is optimally tuned to the disturbance suppression according to Section \ref{sec:4:sub:1}, is applied in the parallel form (to simplify a real-time implementation)
\begin{equation}\label{eq:PIDparallel}
u_\mathrm{pid}(t) = k_\mathrm{p} e(t) + k_\mathrm{i} \int e(t)dt +
k_\mathrm{d} \dot{e}(t),
\end{equation}
where $e(t) = x_{\textmd{ref}}(t) - x(t)$ is the controlled output error.
The proportional, integral, and derivative gains
\begin{equation}
\label{eq:PIDparams}
k_\mathrm{p} = k \frac{T_\mathrm{I} + \tau}{T_\mathrm{I}} =
429, \; k_\mathrm{i} = \frac{k}{T_\mathrm{I}} =
4348, \; k_\mathrm{d} = k \tau = 2.67
\end{equation}
are determined as in Section \ref{sec:4:sub:1}. For the observer design, the assigned parameters (cf. Section \ref{sec:3:sub:4}) $\rho = 4$ and  $\omega_{co} = 2 \pi \times 40 $ rad/sec are determined by the experimental tuning. The measured output error of both evaluated controllers, the PID control \eqref{eq:PIDparallel}, \eqref{eq:PIDparams} and the same PID control augmented by the friction observer so that $u = u_\mathrm{pid} + K^{-1} \tilde{w}_3$, are shown in Figure \ref{fig:experror} (a) and (b) for the positioning and tracking references, respectively.  
\begin{figure}[!h]
\centering
\includegraphics[width=0.98\columnwidth]{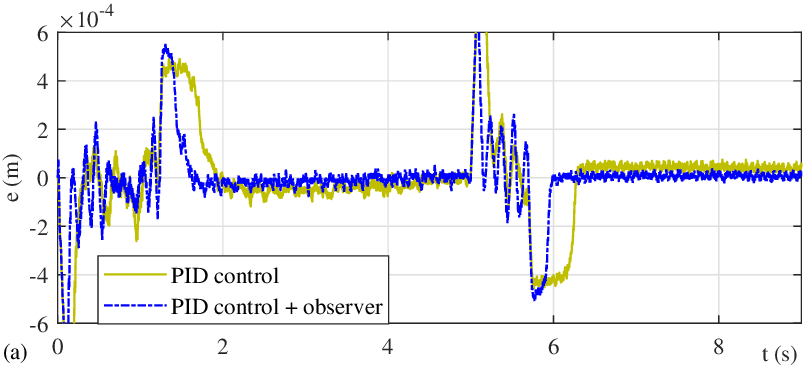}
\includegraphics[width=0.98\columnwidth]{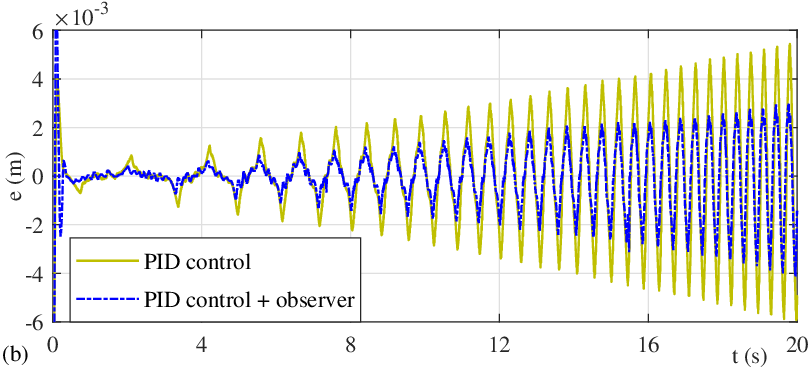}
\caption{Measured control error of positioning (a) and tracking (b).}
\label{fig:experror}
\end{figure}
It can be seen that during the positioning task, the observer-based friction compensation shortens the transient response and allows to reach zero steady-state error, while the latter is not achieved by PID control within both steady-state periods. 
\begin{figure}[!h]
\centering
\includegraphics[width=0.98\columnwidth]{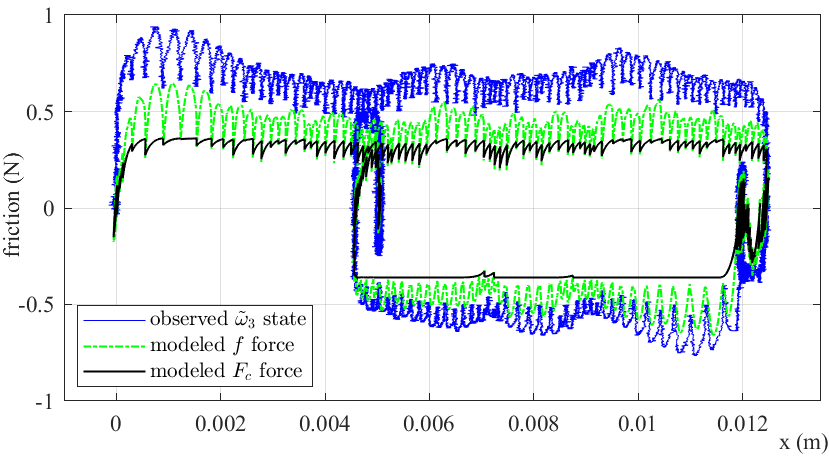}
\caption{Estimated by observer and model-computed friction force.}
\label{fig:frictionloops}
\end{figure}
Also in case of the reference tracking, an improvement through the observer-based friction compensation is clearly visible for all frequencies of the chirp reference signal. For providing an illustrative insight into the observer-based estimate of the nonlinear friction, the latter is shown in Figure \ref{fig:frictionloops} for positioning experiment, in terms of $\tilde{w}_3(t)$ over the measured relative displacement $x(t)$ in comparison with the $f(t)$-value computed solely by the corresponding friction model. Note that here (i.e., when computing the curves shown in Figure \ref{fig:frictionloops}), the same low-pass filtering with $\omega_{co}$, cf. \eqref{eq:3:4:5}, is used for all observed state variables. This is for obtaining the shown $\tilde{w}_3$-value as well as for using the $\tilde{w}_2$-state, which is then the input signal for a model-based computation of the friction force $f$. Note that also the only model-computed Coulomb friction part $F_c(t)$ is shown. It allows an appropriate comparison between the Coulomb and viscous friction contributions to the overall friction force $f$, cf. \eqref{eq:2:1:1}.

\section{Conclusions}
\label{sec:5}

The paper provides the analysis and experimental evaluation of an asymptotic friction observer, proposed \textcolor[rgb]{0.00,0.00,0.00}{initially} in \cite{ruderman2023robust}, whose goal is to estimate the nonlinear state of kinetic friction in addition to \textcolor[rgb]{0.00,0.00,0.00}{the non-measured relative velocity}. The observer benefits from the regular state-space form of describing the unavailable dynamic motion states and hence forms a reduced-order Luenberger observer, while the only measured relative displacement is subject to the sensor noise. We proved, by using the Ignatyev stability theorem \cite{Ignatyev1997}, that the resulting observer with time-varying system matrix is uniformly asymptotically stable. Moreover, we derived the necessary and sufficient conditions for the observation gain parameters to ensure a monotonic convergence of the state variables. It is noteworthy that the proposed convergence analysis leads to only one free design parameter, while the time-dependent eigenvalues remain always real in both the pre-sliding and gross sliding friction regimes. A standard PID feedback regulator was assumed as a reference control system optimized for suppressing the disturbances, \textcolor[rgb]{0.00,0.00,0.00}{since the frictional force can be regarded as a matched disturbance in the motion dynamics.} An experimental comparison of the PID control with and without the friction observer was performed on a laboratory setup for two characteristic positioning and tracking reference trajectories.

\bibliographystyle{elsarticle-harv}
\bibliography{references}

\end{document}